\documentclass[11pt,reqno]{amsart}
\usepackage{amsmath,amsthm,amssymb,bm,dsfont,braket,esint,nicefrac}
\usepackage[mathscr]{eucal}
\usepackage{amsaddr}
\usepackage{upgreek}
\usepackage[left=2cm,right=2cm,top=2cm,bottom=2cm]{geometry}
\usepackage{mathtools}\mathtoolsset{centercolon}
\usepackage[nospace,noadjust]{cite}

\usepackage{pifont} % checkmark
\usepackage[shortlabels]{enumitem}
\usepackage{setspace}
\usepackage{graphicx}
\usepackage{verbatim}
\usepackage{float}
\usepackage{placeins}
\usepackage{array}
\usepackage{booktabs}
\usepackage{multicol,multirow,tabularx}
\usepackage{threeparttable}
\usepackage[update,prepend]{epstopdf}
\usepackage{amsfonts,amssymb,dsfont,xfrac,comment}
\usepackage[abs]{overpic}
\makeatletter
\def\adl@drawiv#1#2#3{%
	\hskip0
	\tabcolsep
	\xleaders#3{#2 0\@tempdimb #1{1}#2 0.5\@tempdimb}%
	#2\z@ plus1fil minus1fil\relax
	\hskip0\tabcolsep}
\newcommand{\cdashlinelr}[1]{%
	\noalign{\vskip\aboverulesep
		\global\let\@dashdrawstore\adl@draw
		\global\let\adl@draw\adl@drawiv}
	\cdashline{#1}
	\noalign{\global\let\adl@draw\@dashdrawstore
		\vskip\belowrulesep}}
\makeatother

\usepackage[dvipsnames]{xcolor}
\usepackage{colortbl}
\usepackage{arydshln}
\usepackage[linktocpage=true,pagebackref]{hyperref}
\hypersetup{
	unicode=false,          % non-Latin characters in Acrobat’s bookmarks
	pdftoolbar=true,        % show Acrobat’s toolbar?
	pdfmenubar=true,        % show Acrobat’s menu?
	pdffitwindow=false,     % window fit to page when opened
	pdfstartview={FitH},    % fits the width of the page to the window
	pdfnewwindow=true,      % links in new PDF window
	colorlinks=true,        % false: boxed links; true: colored links
	linkcolor=Red,          % color of internal links (change box color with linkbordercolor)
	citecolor=ForestGreen,  % color of links to bibliography
	filecolor=Magenta,      % color of file links
	urlcolor=BlueViolet,    % color of external links
}
\renewcommand*{\backref}[1]{}
\renewcommand*{\backrefalt}[4]{%
    \ifcase #1%
          \or [Cited on page~#2.]%
          \else [Cited on pages~#2.]%
    \fi%
    }

\usepackage{doi}
\usepackage{url}
\usepackage{caption, subcaption}
\usepackage{enumitem}
\usepackage{shadethm}
\usepackage{tikz,pgf}

\usepackage{cleveref}
\makeatletter
\ifx\@NODS\undefined%

\let\mathbb=\mathds
\else%
\fi
\makeatother

\DeclareMathOperator*{\argmax}{arg\,max}
\DeclareMathOperator{\Tr}{Tr}

\DeclareMathOperator{\I}{\mathbf{1}}
\DeclareMathOperator{\supp}{supp}

\def\M{\mathsf{M}}

\def\A{\mathsf{A}}
\def\B{\mathsf{B}}

\def\R{\mathsf{R}}

\DeclareMathOperator{\Exp}{\mathbb{E}}

\newcommand{\be}{{\mathbf e}}

\def\0{{\mathbf{0}}}
\def\1{{\mathbf{1}}}
\def\2{{\mathbf{2}}}
\def\3{{\mathbf{3}}}
\def\4{{\mathbf{4}}}
\def\5{{\mathbf{5}}}
\def\6{{\mathbf{6}}}

\def\7{{\mathbf{7}}}
\def\8{{\mathbf{8}}}
\def\9{{\mathbf{9}}}

\def\be{\begin{equation}}
\def\ee{\end{equation}}
\def\bea{\begin{eqnarray}}
\def\eea{\end{eqnarray}}
\def\eps{\varepsilon}

\newcommand{\id}{\operatorname{id}}

\theoremstyle{plain}
\newshadetheorem{theo}{Theorem}[section]
\newshadetheorem{prop}[theo]{Proposition}%[section]
\newshadetheorem{lemm}[theo]{Lemma}%[section]
\theoremstyle{definition}
\theoremstyle{remark}
\newtheorem{remark}[theo]{Remark}%[section]

\makeatletter
\newcommand{\opnorm}{\@ifstar\@opnorms\@opnorm}
\newcommand{\@opnorms}[1]{%
	$\left|\mkern-1.5mu\left|\mkern-1.5mu\left|
	#1
	\right|\mkern-1.5mu\right|\mkern-1.5mu\right|$
}
\newcommand{\@opnorm}[2][]{%
	\mathopen{#1|\mkern-1.5mu#1|\mkern-1.5mu#1|}
	#2
	\mathclose{#1|\mkern-1.5mu#1|\mkern-1.5mu#1|}
}
\makeatother

\usepackage{tikz}
\usetikzlibrary{arrows.meta} % to control arrow size
\tikzset{>={Latex[length=4,width=4]}} % for LaTeX arrow head
\usetikzlibrary{calc}

\colorlet{mylightblue}{blue!5!white}
\colorlet{mydarkblue}{blue!30!black}
\colorlet{myblue}{blue!50!black}
\colorlet{myred}{red!50!black}
\colorlet{mydarkred}{red!30!black}
\colorlet{mydarkgreen}{green!30!black}

\newcommand{\sh}{\kern-0.08em$^\textbf{\#}$\hspace{-3pt}}
\renewcommand{\b}{\kern-0.06em$\flat$}

\newcommand{\Env}{\mathsf E}
\newcommand{\Mh}{{\widehat{\mathsf M}}}
\newcommand{\N}{\mathscr N}
\newcommand{\D}{\mathscr D}
\renewcommand{\eps}{\varepsilon_{\rm EG}}
\newcommand{\avgeps}{\overline{\varepsilon}_{\rm EG}}
\newenvironment{theorem}{\begin{theo}}{\end{theo}}
\theoremstyle{plain}

\allowdisplaybreaks[1]
\hypersetup{pdftitle={Quantum communication: a discrimination approach},pdfauthor={Hao-Chung Cheng}}

\begin{document}

%%%%% How to disable amsart.cls to capitalize the article title?
\let\origmaketitle\maketitle
\def\maketitle{
	\begingroup
	\def\uppercasenonmath##1{} % this disables uppercasing title
	\let\MakeUppercase\relax % this disables uppercasing authors
	\origmaketitle
	\endgroup
}
%%%%%%%%%%%%%%%%%%%%%%%%%%%%%%%%

\title{\bfseries \Large{ 
Quantum communication: a discrimination approach
}}

\author{ \normalsize 
{Hao-Chung Cheng}$^{1\text{--}5}$
}
\address{\small  	
$^1$Department of Electrical Engineering and Graduate Institute of Communication Engineering,\\ National Taiwan University, Taipei 106, Taiwan (R.O.C.)\\
$^2$Department of Mathematics, National Taiwan University\\
$^3$Center for Quantum Science and Engineering, National Taiwan University\\
$^4$Hon Hai (Foxconn) Quantum Computing Center, New Taipei City 236, Taiwan (R.O.C.)\\
$^5$Physics and Mathematics Divisions, National Center for Theoretical Sciences, Taipei 10617, Taiwan (R.O.C.)\\
}

\email{\href{mailto:haochung.ch@gmail.com}{haochung.ch@gmail.com}}

\date{\today}

\begin{abstract}
We develop a state discrimination approach to unassisted quantum communication through a noisy quantum channel.
Starting from the operational duality between entanglement recovery and state discrimination, we construct entanglement-generation codes with explicit decoders constructed from Petz recovery maps. The resulting error bounds hold for every number of channel uses and are expressed in terms of Petz–Rényi coherent information. For product sources, they yield exponential error decay at rates below coherent information.
%We also derive binary-testing bounds and extend the results to entanglement transmission.
\end{abstract}

\maketitle
%\vspace{-2.8em}
%\tableofcontents

%\newpage

\section{Introduction}

Quantum communication aims to transfer a quantum system from Alice to Bob while preserving entanglement with a reference system. Communication is reliable if coding over increasingly many uses of a noisy channel allows Bob to recover the original joint state of the transmitted system and the reference with arbitrarily high accuracy. The quantum capacity is the highest asymptotic transmission rate, in qubits per channel use, at which this is possible \cite{Schumacher96, BarnumKnillNielsen00, Lloyd,ShorCapacity,Devetak}.
An equivalent asymptotic task is \emph{entanglement generation} \cite{Devetak}, in which Alice and Bob use the channel to establish maximally entangled states. This task highlights two complementary perspectives on quantum coding: Bob must be able to recover entanglement with the system retained by Alice, while that system must be nearly uncorrelated with the environment.
In this paper, we study the operational connection between entanglement recovery, state discrimination, and decoupling from the environment. We use this connection to develop a one-shot state discrimination approach to entanglement generation and derive explicit coding bounds.

The quantum capacity theorem, commonly called the Lloyd--Shor--Devetak (LSD) theorem, was established through the work of Lloyd, Shor, and Devetak \cite{Lloyd,ShorCapacity,Devetak}. 
Its central information quantity is the coherent information \cite{SNCoherent}.
For a channel $\N_{\A\to\B}$, a pure reference-input state $\phi_{\R\A}$, and the channel output $\rho_{\R\B}=(\id_\R\otimes\N)(\phi_{\R\A})$, this quantity is
\[
 I(\R\rangle\B)_\rho \coloneqq H(\B)_\rho-H(\R\B)_\rho,
\]
where $H(\B)_\rho \coloneqq - \Tr[\rho_{\B} \log \rho_{\B}]$ is the von Neumann entropy. For a pure channel output state $\rho_{\R\B\Env}$ including the environment, $I(\R\rangle\B)_\rho = H(\B)_\rho-H(\Env)_\rho$. Every rate below this coherent information is asymptotically achievable. Optimizing over inputs to arbitrarily large blocks gives the capacity formula
\begin{equation}
 Q(\N)=\lim_{n\to\infty}\frac1n
 \max_{\phi_{\R\A^n}}
 I(\R\rangle\B^n)_{(\id_{\R}\otimes\N^{\otimes n})(\phi)}.
 \label{eq:intro-capacity}
\end{equation}
The maximization is over pure reference-input states, with a reference that may be entangled with Alice's whole input system $\A^n$; namely, the regularization permits correlations across channel uses.

The proofs of the LSD theorem reveal different mechanisms behind this formula.
Lloyd's construction is based on an asymptotic \emph{packing} argument at the receiver: typical output subspaces associated with different codewords become approximately orthogonal \cite{Lloyd}. For $n$ channel uses, the dimension-counting intuition is
\[
 \underbrace{d}_{\text{codewords}} \cdot
 \;\underbrace{2^{nH(\Env)_{\rho}}}_{\text{output subspace per codeword}}
 \ \lesssim\
 \underbrace{2^{nH(\B)_{\rho}}}_{\text{available output space}},
\]
which suggests rates $\frac1n\log_2d<H(\B)_{\rho}-H(\Env)_{\rho}$. Shor developed a related packing argument using a pretty-good measurement \cite{Belavkin75, HausladenWootters94} to distinguish output vectors indexed by a codeword and a typical Kraus error caused by $\mathcal{N}$, and constructs a coherent recovery operation from the measurement vectors \cite{ShorCapacity}.
A later proof by Hayden, Shor, and Winter combined distinguishability in two Fourier-conjugate bases with an information-uncertainty relation \cite{ChristandlWinter} to establish decoupling and hence the existence of a quantum decoder~\cite{HSWGaussian}.

Devetak's proof proceeds through private classical communication \cite{Devetak}. A packing argument enables the receiver to distinguish a message together with a randomization label chosen locally by Alice. 
A \emph{covering} argument, based on the operator concentration bound of Ahlswede and Winter \cite{AhlswedeWinter}, ensures that averaging over this label makes the environmental output nearly independent of the message.
A coherent implementation of the private coding construction then yields entanglement generation. Thus, receiver distinguishability and environmental independence already play complementary roles in the capacity theorem.

Hayden, Horodecki, Winter, and Yard developed a decoupling proof of the quantum coding theorem \cite{HHWY}, building on earlier work on approximate quantum error correction and quantum state merging \cite{SchumacherWestmoreland02,HorodeckiOppenheimWinter}. Their construction uses a unitarily covariant subspace measurement to generate random codes and shows that the retained reference is approximately maximally mixed and uncorrelated with the environment. Uhlmann's theorem then guarantees the existence of a decoding operation on Bob's system that produces a nearly maximally entangled state shared with the retained reference \cite{Uhlmann}. Thus, a decoupling estimate suffices to guarantee reliable entanglement generation without explicitly constructing the decoder. Together with related random-code analyses such as Klesse's \cite{Klesse}, this established decoupling as a direct approach to quantum coding.

Decoupling subsequently became a general one-shot framework for quantum information processing \cite{DupuisThesis,DBWR,SharpCovering,BCY}, with further developments in quantitative bounds and error exponents \cite{LiYaoDecoupling,CDG}.
Related covering and resolvability methods have been developed for quantum soft covering \cite{CGSoftCovering,SGCSecondOrder,APWCovering}, privacy amplification \cite{RennerThesis,HayashiPA2011,HayashiPA2013,HayashiPA2014,HayashiPA2015,DupuisNoSmoothing,RegulaTomamichel}, and private communication \cite{HayashiResolvability2006,HayashiWiretap2015}.
Convex-split methods \cite{ADJConvexSplit,CGConvexSplit,SharpCovering} have also found applications in entanglement-assisted communication \cite{AJWEA,AnshuJain} and quantum channel simulation \cite{LiYaoSimulation,BCGBroadcast}.
Recently, Berta, Cheng, and Yao \cite{BCY} developed relative-entropy decoupling and derived achievable error exponents for quantum communication in purified distance, expressed in terms of Petz--R\'enyi coherent information (see also \cite{Yamasaki} for decoupling beyond finite dimensions).
These approaches share a common principle: suitable randomization makes the relevant state approximate a prescribed target state, paralleling classical soft covering and channel resolvability, where a mixture of codeword output distributions approximates a target distribution \cite{HanVerdu1993,HayashiResolvability2006,CuffChannelSynthesis,Hayashi2017}.
In privacy amplification, convex splitting, and decoupling tasks, the target is a product state, so this approximation establishes that the protected system is nearly independent of the side information.

A second line of work analyzes an explicit decoder and brings quantum coding closer to state discrimination. Beigi, Datta, and Leditzky \cite{BDL} extended the pretty-good-measurement analysis of classical-quantum (c-q) coding to entanglement transmission, replacing the measurement by a Petz recovery map and decoding success by fidelity \cite{BeigiGohari}.
Their method yields one-shot and second-order coding bounds. 
Burri \cite{Burri} gave explicit fidelity expressions for Petz and rotated Petz decoders, including a characterization through the complementary channel. More recently, Can and Wilde \cite{CW} used a trace inequality from the present author's paper \cite{ChengSimple} to control the infidelity of a Petz decoder by a binary-testing cost, obtaining one-shot bounds for entanglement generation and one-way entanglement distillation.
These results motivate studying quantum communication directly through the discrimination methods that govern classical-message decoding.

To explain why the decoupling and discrimination approaches both work, consider a pure preparation $\psi_{\M\A}$ at the encoder and a pure output $\theta_{\M\B\Env}$ obtained by applying a Stinespring isometry of the channel. Alice retains the quantum system $\M$, Bob receives $\B$, and the target is a maximally entangled state $\Phi_{\M\Mh}$ of Schmidt rank $d=|\M|$, containing $\log_2d$ ebits.
Let $\{W_z\}_{z=1}^{d^2}$ be the Weyl unitaries on $\M$. The known operational dualities collected in Theorem~\ref{thm:duality} relate the following three optimal errors:\footnote{Equation \eqref{eq:intro-duality} can be written in terms of the success probability and squared fidelity. We choose this representation since this paper focuses on error analysis.}
\begin{equation}
\begin{aligned}
 \min_{\sigma_\Env}
 P^2\!\left(\theta_{\M\Env},\frac{I_\M}{d}\otimes\sigma_\Env\right)
 &\overset{\text{\cite{KRS}}}{=}
 \min_{\D_{\B\to\Mh}}
 P^2\!\left((\id_\M\otimes\D)(\theta_{\M\B}),\Phi_{\M\Mh}\right)
 \\
 &\overset{\text{\cite{Jencova}}}{=}
 P_{\rm e}^{\star}\!\left(
 \left\{\frac1{d^2},
 (W_z\otimes I_\B)\theta_{\M\B}(W_z^\dagger\otimes I_\B)
 \right\}_{z=1}^{d^2}\right).
\end{aligned}
\label{eq:intro-duality}
\end{equation}
Here $P^2=1-F^2$ is squared purified distance, $F$ is root fidelity, and $P_{\rm e}^{\star}$ is the minimum average discrimination error of the quantum ensemble.
The minimum over $\sigma_\Env$ is over density operators, and the minimum over $\D_{\B\to\Mh}$ is over decoding channels.
The first equality follows from the operational characterizations of conditional min- and max-entropies by K\"onig, Renner, and Schaffner \cite{KRS}.
The second follows from Jen\v{c}ov\'a's characterization of the optimal recovery fidelity through discrimination of the Weyl-rotated ensemble \cite{Jencova}.
In the classical-quantum setting, Renes established related dualities between privacy amplification and data compression with quantum side information \cite{RenesPADuality}, and between discrimination on a channel and decoupling on its dual \cite{RenesChannelDuality}.
These ideas also yield error-exponent bounds for classical-quantum channels \cite{Renes2025}.
Connections also appear in the duality of teleportation and dense coding \cite{ChL}.
Section~\ref{sec:duality} gives a self-contained proof of the two equalities above, which we use as the starting point for our coding bounds.

Equation~\eqref{eq:intro-duality} gives a precise operational connection between the covering and packing perspectives, illustrated in Figure~\ref{fig:operational-duality}.
From the covering perspective, entanglement can be recovered accurately when $\M$ is nearly maximally mixed and decoupled from the environment, as illustrated in \cite{SchumacherWestmoreland02, HHWY}.
From the packing perspective, accurate recovery is equivalent to distinguishing the $d^2$ equiprobable Weyl-rotated states on $\M\B$ with small error.
The joint measurement on $\M\B$ serves as an auxiliary task in the analysis; the corresponding decoder acts only on Bob's system $\B$.
The two perspectives thus characterize the same optimal entanglement-recovery error.

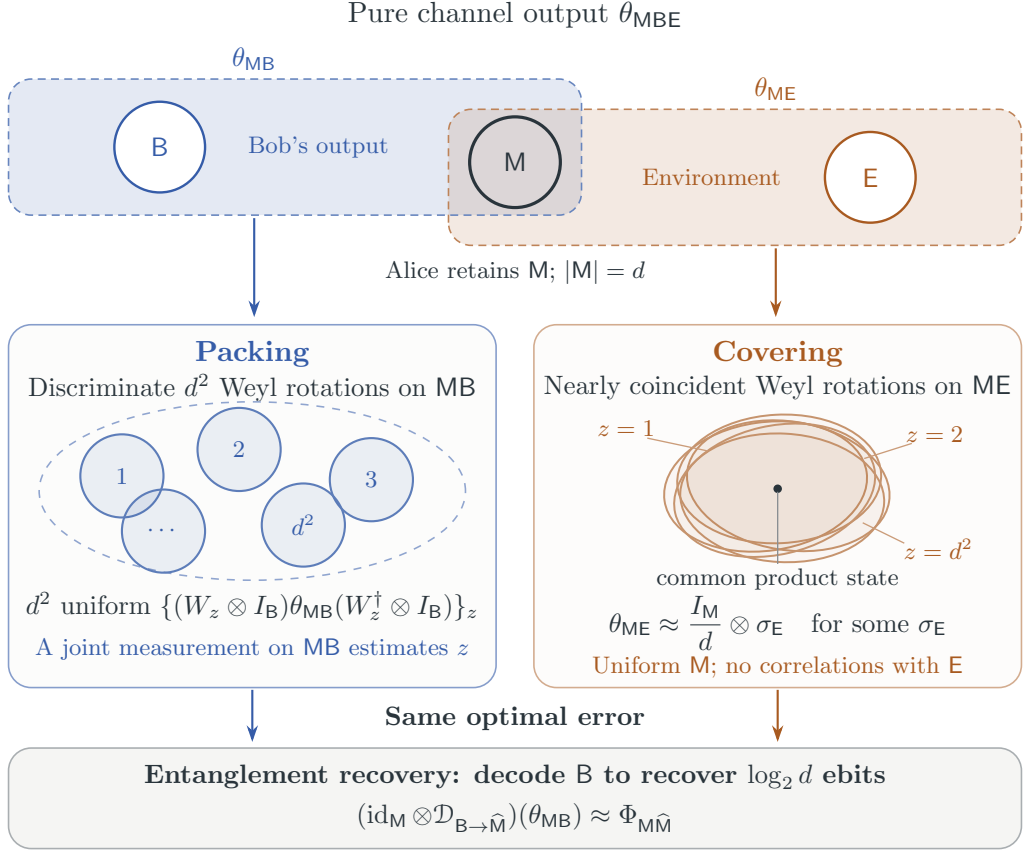
\begin{figure}[tbp]
\centering
\begin{tikzpicture}[x=1cm,y=1cm,
  font=\small,
  system/.style={circle,minimum size=12mm,inner sep=0pt,line width=.9pt},
  flow/.style={-{Stealth[length=2.2mm,width=1.5mm]},line width=.8pt},
  panel/.style={rounded corners=3mm,line width=.6pt},
  every node/.style={align=center}]
 \definecolor{packingblue}{HTML}{355CA8}
 \definecolor{coveringorange}{HTML}{A85A20}
 \definecolor{dualityink}{HTML}{29343B}
 \definecolor{panelpaper}{HTML}{FCFCFB}
 \definecolor{recoverypaper}{HTML}{F5F5F3}

 % Dashed marginal frames: neither frame denotes a joint local laboratory.
 % Draw translucent fills first, then both outlines; their overlap contains M.
 \node[font=\normalsize,text=dualityink] at (0,1.95)
   {Pure channel output $\theta_{\M\B\Env}$};
 \path[rounded corners=3mm,fill=packingblue,fill opacity=.13]
   (-6.7,-.70) rectangle (.88,1.10);
 \path[rounded corners=3mm,fill=coveringorange,fill opacity=.13]
   (-.88,-1.10) rectangle (6.7,.70);
 \draw[panel,draw=packingblue!75,dash pattern=on 3pt off 2pt]
   (-6.7,-.70) rectangle (.88,1.10);
 \draw[panel,draw=coveringorange!75,dash pattern=on 3pt off 2pt]
   (-.88,-1.10) rectangle (6.7,.70);
 \node[text=packingblue,font=\small] at (-3.45,1.37) {$\theta_{\M\B}$};
 \node[text=coveringorange,font=\small] at (3.45,.97) {$\theta_{\M\Env}$};
 \node[system,draw=packingblue,fill=white,text=packingblue] at (-4.7,.20) {$\B$};
 \node[system,draw=dualityink,fill=none,text=dualityink,line width=1.2pt] at (0,0) {$\M$};
 \node[system,draw=coveringorange,fill=white,text=coveringorange] at (4.7,-.20) {$\Env$};
 \node[text=packingblue,font=\footnotesize] at (-2.6,.20) {Bob's output};
 \node[text=coveringorange,font=\footnotesize] at (2.6,-.20) {Environment};
 \node[inner sep=0pt,font=\footnotesize,text=dualityink] at (0,-1.45)
   {Alice retains $\M$; $|\M|=d$};

 \draw[flow,packingblue] (-3.45,-.74)--(-3.45,-2.05);
 \draw[flow,coveringorange] (3.45,-1.14)--(3.45,-2.05);

 \begin{scope}[yshift=-5mm]
 % Packing: touching and overlapping contours allow nonzero error.
 \draw[panel,draw=packingblue!60,fill=panelpaper]
   (-6.7,-1.65) rectangle (-.25,-6.45);
 \node[text=packingblue,font=\normalsize\bfseries] at (-3.475,-2.04) {Packing};
 \node[text=dualityink] at (-3.475,-2.48) {Discriminate $d^2$ Weyl rotations on $\M\B$};
 \draw[packingblue!55,dashed,line width=.6pt]
   (-3.475,-3.85) ellipse [x radius=2.82,y radius=1.18];
 \foreach \x/\y/\lab in {-5.2/-3.65/1,-3.65/-3.3/2,-1.9/-3.7/3,-4.65/-4.38/{\cdots},-2.8/-4.3/{d^2}} {
   \draw[draw=packingblue!80,fill=packingblue,fill opacity=.09,line width=.8pt]
     (\x,\y) circle [radius=.55];
   \node[text=packingblue,font=\footnotesize] at (\x,\y) {$\lab$};
 }
 \node[text=dualityink] at (-3.475,-5.37)
   {$d^2$ uniform $\{(W_z\otimes I_\B)\theta_{\M\B}(W_z^\dagger\otimes I_\B)\}_z$};
 \node[text=packingblue,font=\footnotesize] at (-3.475,-5.94)
   {A joint measurement on $\M\B$ estimates $z$};

 % Decoupling: the ME orbit states are close to a common product reference.
 \draw[panel,draw=coveringorange!50,fill=coveringorange!2]
   (.25,-1.65) rectangle (6.7,-6.45);
 \node[text=coveringorange,font=\normalsize\bfseries] at (3.475,-2.04)
   {Covering};
 \node[text=dualityink] at (3.475,-2.48) {Nearly coincident Weyl rotations on $\M\Env$};
 \begin{scope}[shift={(3.475,-3.82)}]
  \foreach \x/\y/\ang in {-.16/.03/8,.14/.04/-9,.02/.13/3,.06/-.12/-3,-.04/-.02/13} {
   \draw[draw=coveringorange!65,fill=coveringorange,fill opacity=.035,line width=.75pt,
       shift={(\x,\y)},rotate=\ang]
      (0,0) ellipse [x radius=1.35,y radius=.85];
  }
  \fill[dualityink] (0,0) circle [radius=1.6pt];
 \end{scope}
 \node[text=coveringorange,font=\footnotesize] at (1.47,-3.02) {$z=1$};
 \draw[coveringorange!60,line width=.5pt] (1.75,-3.15)--(2.55,-3.32);
 \node[text=coveringorange,font=\footnotesize] at (5.55,-3.1) {$z=2$};
 \draw[coveringorange!60,line width=.5pt] (5.25,-3.22)--(4.62,-3.42);
 \node[text=coveringorange,font=\footnotesize] at (5.6,-4.62) {$z=d^2$};
 \draw[coveringorange!60,line width=.5pt] (5.15,-4.54)--(4.6,-4.27);
 \draw[dualityink!55,line width=.5pt] (3.475,-3.86)--(3.475,-4.91);
 \node[fill=coveringorange!2,inner sep=1.5pt,text=dualityink,font=\footnotesize]
   at (3.475,-5.02) {common product state};
 \node[text=dualityink] at (3.475,-5.59)
   {$\displaystyle\theta_{\M\Env}\approx\frac{I_\M}{d}\otimes\sigma_\Env\quad\text{for some }\sigma_\Env$};
 \node[text=coveringorange,font=\footnotesize] at (3.475,-6.18)
   {Uniform $\M$; no correlations with $\Env$};

 % The only text between the panels and the recovery box is the equality label.
 \draw[flow,packingblue] (-3.475,-6.48)--(-3.475,-7.17);
 \draw[flow,coveringorange] (3.475,-6.48)--(3.475,-7.17);
 \node[text=dualityink,font=\small\bfseries] at (0,-6.86)
   {Same optimal error};
 \draw[panel,draw=dualityink!40,fill=recoverypaper]
   (-6.7,-7.25) rectangle (6.7,-8.58);
 \node[text=dualityink,font=\small\bfseries] at (0,-7.63)
   {Entanglement recovery: decode $\B$ to recover $\log_2 d$ ebits};
 \node[text=dualityink] at (0,-8.17)
   {$(\id_\M\otimes\D_{\B\to\Mh})(\theta_{\M\B})\approx\Phi_{\M\Mh}$};
 \end{scope}
\end{tikzpicture}
\caption{Operational duality for the pure channel output $\theta_{\M\B\Env}$, with $\M$ common to the two marginals.
The contours schematically represent Weyl-rotated states; $W_z$ acts only on $\M$. The packing contours may overlap, allowing nonzero discrimination error. On $\M\Env$, all Weyl-rotations should remain close to the same product state $I_\M/d\otimes\sigma_\Env$ for some environment state $\sigma_\Env$, and the recovery channel gives the same error, as in \eqref{eq:intro-duality}. The joint discrimination measurement is auxiliary; the communication decoder acts only on $\B$.}
\label{fig:operational-duality}
\end{figure}

In this paper, we use a one-shot state discrimination approach to extend recently developed packing techniques to quantum communication \cite{ChengSimple,BeigiTomamichel24,CL}.
We construct explicit decoders through ordinary Petz recovery
\cite{CW} and averaged rotated Petz recovery of an auxiliary Weyl
ensemble, the latter corresponding to logarithmic-derivative
measurements \cite{BeigiTomamichel24,CL}.
The resulting coding bounds hold for every number $n$ of channel uses.
Already for product sources, they yield exponential decay of the entanglement-generation error in squared purified distance at every rate below the corresponding coherent information, with an achievable exponent expressed in terms of Petz--R\'enyi coherent information, as in \cite{BCY}.\footnote{Earlier typical-subspace proofs
	\cite{HorodeckiOppenheimWinter,HHWY}
	already establish exponential error decay for sufficiently
	large blocklengths $n$, with the decay exponent left implicit.}
Specifically, let $n\geq1$, let $\phi_{\R\A^n}$ be a pure source state, and write
$\rho_{\R\B^n}=(\id_\R\otimes\N^{\otimes n})(\phi_{\R\A^n})$.\footnote{Our one-shot theorem also applies to products of different channels and to general correlated block channels.}
Optimizing the construction over sources gives a code of target Schmidt rank $d_n$ satisfying
\begin{equation}
 \eps(d_n)\leq
 2^{-\sup_{\alpha\in[\sfrac12,1)}\frac{1-\alpha}{\alpha}
 \left[\sup_{\phi_{\R\A^n}} I_\alpha(\R\rangle\B^n)_\rho-\log_2 d_n\right]},
 \qquad n\geq1.
 \label{eq:intro-exponent}
\end{equation}
Here $I_\alpha$ is the Petz--R\'enyi coherent information defined in Section~\ref{sec:coding}.
Continuity at $\alpha=1$ makes the exponent positive whenever $\log d_n <I(\R\rangle\B^n)_\rho$; for product sources, additivity yields exponential decay in $n$ at every rate below the corresponding coherent information.
Optimizing over correlated block inputs may improve the achievable error exponent at rates below the quantum capacity.\footnote{Similar phenomenon has been shown for entanglement-assisted communication \cite{ChengBerta26}.}
The achievability bound mirrors the quantum Arimoto converse underlying the recent strong-converse results \cite{ChengTomamichel26,BeigiTomamichel26}:
%For any code of target dimension $d_n$, with encoded channel output $\theta_{\M\B^n}$ before decoding, that converse gives
\begin{equation}
 1-\eps(d_n)\leq
 2^{-\sup_{\alpha>1}\frac{\alpha-1}{\alpha}
 \left[\log_2d_n-\sup_{\phi_{\R\A^n}} \widetilde I_\alpha(\R\rangle\B^n)_\rho\right]},
 \qquad n\geq1,
 \label{eq:intro-sc}
\end{equation}
where $\widetilde I_\alpha$ is the negative optimized sandwiched conditional R\'enyi entropy \cite{ChengTomamichel26}.

Section~\ref{sec:duality} reviews the operational duality and gives a self-contained proof. Section~\ref{sec:coding} presents the coding theorem together with its construction and error analysis for entanglement generation and transmission.
Section~\ref{sec:small-deviations} shows the binary-testing bound by using the author's trace inequality \cite{ChengSimple}.
Section~\ref{sec:consequences} concludes this paper.
Appendix~\ref{app:auxiliary} collects the tilting and twirling statements; Appendix~\ref{app:petz} explains the Petz recovery correspondence.

\section{Entanglement generation and operational duality} \label{sec:duality}

All systems are finite dimensional, and $|\M|$ denotes the dimension of $\M$. Channels are completely positive, trace-preserving linear maps; a star `$*$' denotes their trace adjoint. A POVM is a family of positive operators that sums to the identity. For normalized states, write
\[
 F(\rho,\sigma)=\|\sqrt{\rho}\sqrt{\sigma}\|_1,
 \qquad P^2(\rho,\sigma)=1-F^2(\rho,\sigma).
\]
Thus $F$ is root fidelity and $P$ is purified distance. Let $|\M|=|\Mh|=d$ and
\[
 \ket{\Phi}_{\M\Mh}=\frac1{\sqrt d}\sum_{m=0}^{d-1}\ket m_\M\ket m_\Mh,
 \qquad \Phi_{\M\Mh}=\ket{\Phi}\bra{\Phi}_{\M\Mh}.
\]
An entanglement-generation code for a channel $\N_{\A\to\B}$ is a pair consisting of a normalized pure preparation $\psi_{\M\A}$ and a decoding channel $\D_{\B\to\Mh}$. Alice retains $\M$ and sends $\A$ through $\N$; Bob applies $\D$. Its entanglement-generating error is
\begin{align}
 \eps(\psi,\D;\N)
 &\coloneqq P^2\!\left((\id_\M\otimes(\D\circ\N))(\psi_{\M\A}),\Phi_{\M\Mh}\right)\notag\\
 &=1-\Tr\!\left[\Phi_{\M\Mh}
 (\id_\M\otimes(\D\circ\N))(\psi_{\M\A})\right].
 \label{eq:eg-error}
\end{align}
The preparation need not have a uniform marginal on $\M$. For a specified code of target dimension $d$, we abbreviate its error as $\eps(d)$; for a random ensemble of such codes, $\avgeps(d)$ denotes the expectation of \eqref{eq:eg-error}. The randomness describes code selection, and the selected encoder and decoder are fixed before communication.

Let $\{W_z\}_{z=1}^{d^2}$ be the $d$-dimensional Weyl unitaries, with $W_1=I_\M$. Concretely, these are the shift and phase operators
$\ket m\mapsto e^{2\pi \mathrm{i} vm/d}\ket{m+u\bmod d}$, for $u,v=0,\ldots,d-1$, indexed by $z$. Their conjugations satisfy
\begin{equation}
 \sum_{z=1}^{d^2}(W_z\otimes I_\B)X_{\M\B}(W_z^\dagger\otimes I_\B)
 =dI_\M\otimes\Tr_\M X_{\M\B}.
 \label{eq:twirl}
\end{equation}
We write $P_{\rm e}^{\star}$ for the minimum average discrimination error over all POVMs.

\begin{theorem}[Operational duality {\cite{KRS, Jencova}}]
\label{thm:duality}
For every normalized pure state $\theta_{\M\B\Env}$,
\begin{equation}
\begin{aligned}
 \min_{\D_{\B\to\Mh}}
 P^2\!\left((\id_\M\otimes\D)(\theta_{\M\B}),\Phi_{\M\Mh}\right)
 &=
 P_{\rm e}^{\star}\!\left(
 \left\{\frac1{d^2},\,
 (W_z\otimes I_\B)\theta_{\M\B}(W_z^\dagger\otimes I_\B)
 \right\}_{z=1}^{d^2}\right)
 \\[1mm]
 &=
 \min_{\substack{\sigma_\Env\geq0\\ \Tr\sigma_\Env=1}}
 P^2\!\left(\theta_{\M\Env},\frac{I_\M}{d}\otimes\sigma_\Env\right).
\end{aligned}
\label{eq:duality}
\end{equation}
The first minimum is over channels from $\B$ to $\Mh$.
\end{theorem}

\begin{proof}
We first establish the constructive decoder-measurement correspondence. A channel $\D_{\B\to\Mh}$ determines
\begin{equation}
 T_{\M\B}=(\id_\M\otimes\D^*)(\Phi_{\M\Mh}),
 \qquad T_{\M\B}\geq0,\qquad
 \Tr_\M T_{\M\B}=\frac{I_\B}{d}.
 \label{eq:choi}
\end{equation}
Its Weyl conjugates
\begin{equation}
 \left\{(W_z\otimes I_\B)T_{\M\B}(W_z^\dagger\otimes I_\B)\right\}_{z=1}^{d^2}
 \label{eq:seed-povm}
\end{equation}
form a POVM by \eqref{eq:twirl}. Each hypothesis and its corresponding test undergo the same conjugation. Consequently, every correct-outcome probability equals
\begin{equation}
 \Tr[\theta_{\M\B}T_{\M\B}]
 =\Tr\!\left[
 \Phi_{\M\Mh}(\id_\M\otimes\D)(\theta_{\M\B})
 \right],
 \label{eq:same-success}
\end{equation}
which is the squared fidelity with the maximally entangled target.

Conversely, from any POVM $\{\Pi_{\M\B}^z\}_{z=1}^{d^2}$ define
\begin{equation}
 T_{\M\B}
 =\frac1{d^2}\sum_{z=1}^{d^2}
 (W_z^\dagger\otimes I_\B)\Pi_{\M\B}^z(W_z\otimes I_\B).
 \label{eq:from-povm}
\end{equation}
It is positive and has partial trace $I_\B/d$. The normalized Choi characterization therefore supplies a channel satisfying \eqref{eq:choi}. Its overlap with $\theta_{\M\B}$ is exactly the given POVM's average success. Explicitly,
\begin{equation}
 \D(\xi_\B)
 =d\left(\Tr_\B[T_{\M\B}(I_\M\otimes\xi_\B)]\right)^{\mathsf T},
 \label{eq:decoder-action}
\end{equation}
where the transpose is in the basis defining $\Phi$, with $\M$ identified with $\Mh$. This proves the first equality.

For the second equality, purify $\sigma_\Env$ and dilate the decoder to an isometry on $\B$, allowing an arbitrarily large auxiliary output. Uhlmann's theorem equates the fidelity between $\theta_{\M\Env}$ and $I_\M/d\otimes\sigma_\Env$ with the maximum overlap of their purifications. A purification of the latter state is $\ket{\Phi}_{\M\Mh}$ tensored with a purification of $\sigma_\Env$. Optimizing its auxiliary vector gives the fidelity of the decoded state with $\Phi_{\M\Mh}$. Optimizing the decoder and $\sigma_\Env$, then squaring and subtracting from one, proves the claim.
\end{proof}

The discrimination measurement is joint on $\M\B$. In that auxiliary problem it can be realized by decoding $\B$ and measuring $\M\Mh$ in the Weyl Bell basis. In entanglement generation, Bob applies only the channel on $\B$; the joint Bell measurement represents the fidelity test. The $d^2$ labels describe the discrimination problem, while the recovered entanglement is $\log_2d$ ebits.

For example, send Bob's half of $\Phi$ through a flagged erasure channel. When no erasure occurs, the orbit consists of orthogonal Bell states. On the erased branch, all $d^2$ hypotheses coincide and the best success probability is $1/d^2$. Both optimal errors therefore equal the erasure probability multiplied by $1-1/d^2$. The discrimination interpretation identifies exactly which part of the state prevents recovery.

%The decoupling expression in \eqref{eq:duality} must be optimized over $\sigma_\Env$: substituting the fixed marginal $\theta_\Env$ gives a sufficient bound, rather than the exact identity.

\section{Error analysis of one-shot entanglement generation}
\label{sec:coding}

Fix a channel $\N_{\A\to\B}$, a normalized input state $\phi_\A$, and a purification $\phi_{\R\A}$. We keep this reference unchanged throughout the proof; it may be minimal, with $|\R|=\operatorname{rank}\phi_\A$.
Write
\begin{align} \label{eq:channel_output}
 \rho_{\R\B}=(\id_\R\otimes\N)(\phi_{\R\A}).
\end{align}
This source channel output $\rho_{\R\B}$ is different from the encoded channel output $\theta_{\M\B}$ in the operational duality, Theorem~\ref{thm:duality}; the encoding construction below obtains the latter from the former.

For $\alpha\in[\sfrac12,1)$, define the Petz--R\'enyi coherent information \cite{Petz1988, SW13, TBH} by
\begin{equation}
 I_\alpha(\R\rangle\B)_\rho
 =-\frac{\alpha}{1-\alpha}
 \log_2\Tr\!\left[\left(\Tr_\R[\rho_{\R\B}^{\alpha}]\right)^{1/\alpha}\right].
 \label{eq:coherent}
\end{equation}
It is invariant under isometries on the purifying reference. 
Throughout the paper, all information quantities and rates use base-two logarithms unless written as $\ln$."

\begin{theorem}[One-shot random coding bound for entanglement generation]
\label{thm:coding}
%For every channel input and purification as above, every integer $d\geq2$, and every $\alpha\in[\sfrac12,1)$, there is an ensemble of pure preparations and explicit decoders whose average channel input is $\phi_\A$. The two quotient decoders introduced in Section~\ref{sec:encoding-ensemble} give the following bounds with the same preparation ensemble.
For $2\leq d\leq|\R|$, $\alpha \in [\sfrac12, 1)$, any channel input $\phi_\A$, and the corresponding channel output $\rho_{\R\B}$ in \eqref{eq:channel_output},
the ensemble error satisfies
\begin{align}
 \avgeps(d)&\leq c_\alpha
 \left(\frac{d|\R|-1}{|\R|^2-1}\right)^{\frac{1-\alpha}{\alpha}}
 \Tr\!\left[\rho_{\R\B}^{\alpha}
 \left(|\R|I_\R\otimes\Tr_\R\rho_{\R\B}^{\alpha}
                   -\rho_{\R\B}^{\alpha}\right)^{\frac{1-\alpha}{\alpha}}\right]
 \label{eq:operator-bound}\\
 &\leq c_\alpha
 \left[\frac{d|\R|-1}{|\R|^2-1}
 \left(|\R|\,2^{-I_\alpha(\R\rangle\B)_\rho}-1\right)\right]^{\frac{1-\alpha}{\alpha}}.
 \label{eq:scalar-bound}
\end{align}
For $d\geq|\R|$, one preparation and decoder satisfy
\begin{align}
 \eps(d)
 &\leq c_\alpha\Tr\!\left[\rho_{\R\B}^{\alpha}
 \left(dI_\R\otimes\Tr_\R\rho_{\R\B}^{\alpha}
       -\rho_{\R\B}^{\alpha}\right)^{\frac{1-\alpha}{\alpha}}\right]
 \label{eq:large-target-operator}\\
 &\leq c_\alpha
 \left(d\,2^{-I_\alpha(\R\rangle\B)_\rho}-1\right)^{\frac{1-\alpha}{\alpha}}.
 \label{eq:large-target-scalar}
\end{align}
For every $d\geq2$, these bounds imply
\begin{equation}
 \avgeps(d)\leq c_\alpha\,
 2^{-\frac{1-\alpha}{\alpha}(I_\alpha(\R\rangle\B)_\rho-\log_2d)}.
 \label{eq:exponential-bound}
\end{equation}
Here, the $c_\alpha \coloneqq \left(\frac{1-\alpha}{\alpha}\right)^{\frac{1-\alpha}{\alpha}}
\left(\frac{2\alpha-1}{\alpha}\right)^{\frac{2\alpha-1}{\alpha}} \in [\sfrac12,1]$ is defined in \eqref{eq:constant}.
%At least one preparation and its decoder attain the ensemble bounds.
\end{theorem}

%The target $d=1$ has zero error. For $d\leq|\R|$ we select a subspace; for larger targets we use the source directly.

Section~\ref{sec:encoding-ensemble} provides the coding construction and we prove \ref{thm:coding} in Section~\ref{sec:proof}.
Section~\ref{sec:tightness} shows that the coefficient $c_{\alpha}$ is sharp.
We discuss entanglement transmission and strong subspace transmission in Section~\ref{sec:transmission}.

\subsection{The encoding ensemble and its discrimination POVMs}
\label{sec:encoding-ensemble}

We use a positive operator quotient to normalize the discrimination measurement. For positive operators $Y$ and $X$ with $\supp (Y)\subseteq\supp (X)$, choose either convention
\begin{equation}
 \frac{Y}{X}\coloneqq
 \begin{cases}
 \displaystyle\mathrm D\ln[X](Y)
 =\int_0^\infty(X+tI)^{-1}Y(X+tI)^{-1}\,\mathrm dt,
 &\text{logarithmic derivative},\\[2mm]
 X^{-\sfrac12}YX^{-\sfrac12},&\text{symmetric quotient},
 \end{cases}
 \label{eq:quotient}
\end{equation}
and keep this choice fixed throughout the construction. Inverses and the integral are restricted to $\supp(X)$ and extended by zero on its orthogonal complement.
Both conventions give $X/X$ equal to the support projection of $X$; the natural logarithm ensures this normalization. Both are positive and linear in the numerator, unitarily covariant, and obey $Y/(rX)=r^{-1}(Y/X)$ for $r>0$.

Suppose $2\leq d\leq|\R|$. Choose finitely many isometries $V_j:\M\to\R$ and positive weights $w_j$ satisfying
\begin{equation}
 \sum_jw_j=1,
 \qquad
 \sum_jw_jV_jV_j^\dagger=\frac d{|\R|}I_\R.
 \label{eq:resolution}
\end{equation}

The random subspace encoding used below is inspired by \cite{HHWY}; see also \cite{HSWGaussian,CW} for related coding constructions. 
The following explicit choice satisfies \eqref{eq:resolution} for every $2\leq d\leq|\R|$.
Fix an orthonormal basis of $\R$, index $j$ by its $d$-element subsets, and let $V_j$ map the standard basis of $\M$ to the corresponding subset. Each basis vector belongs to $\binom{|\R|-1}{d-1}$ subsets, so
\[
 \sum_jV_jV_j^\dagger=\binom{|\R|-1}{d-1}I_\R,
 \qquad
 w_j=\binom{|\R|}{d}^{-1},
 \qquad
 \frac{|\R|w_j}{d}=\binom{|\R|-1}{d-1}^{-1}.
\]
The analysis allows any weighted family satisfying \eqref{eq:resolution}. The subspaces may overlap.

The random code is specified by a unitary $U$ and an index $j$, both known to Alice and Bob before communication. Draw $U$ according to normalized Haar measure on $\R$, and conditionally draw $j$ with probability
\begin{equation}
 p_j(U)=\frac{|\R|}{d}w_j
 \Tr[V_j^\dagger U^\dagger\phi_\R UV_j].
 \label{eq:encoding-probability}
\end{equation}
The encoding for this pair is the pure preparation
\begin{equation}
 \ket{\psi^{U,j}}_{\M\A}
 =\frac{(V_j^\dagger U^\dagger\otimes I_\A)\ket\phi_{\R\A}}
 {\sqrt{\Tr[V_j^\dagger U^\dagger\phi_\R UV_j]}}.
 \label{eq:encoding}
\end{equation}
Alice retains $\M$ and sends $\A$ through the channel; Bob uses the decoder specified below for the same pair $(U,j)$. The weights $w_j$ describe the fixed subspace family, whereas $p_j(U)$ also accounts for the source mass in the rotated subspace. Preparations with $p_j(U)=0$ are omitted, but all indices are retained in operator sums. For every $U$, \eqref{eq:resolution} gives
\begin{equation}
 \sum_jp_j(U)=1,
 \qquad
 \sum_jp_j(U)\psi_\A^{U,j}=\phi_\A.
 \label{eq:average-input}
\end{equation}
Thus the prescribed input is the ensemble average even before averaging over $U$. The random choice describes code selection: a suitable pair $(U,j)$ will be fixed before use, so no shared randomness or transmission of $j$ is required during communication.

The encoded output is
\begin{equation}
 \theta_{\M\B}^{U,j}
 =\frac{(V_j^\dagger U^\dagger\otimes I_\B)\rho_{\R\B}(UV_j\otimes I_\B)}
 {\Tr[V_j^\dagger U^\dagger\phi_\R UV_j]}.
 \label{eq:encoded-output}
\end{equation}
For each encoding, choose the complete POVM for its $d^2$ Weyl hypotheses to be
\begin{equation}
 \left\{
 \frac{(W_zV_j^\dagger U^\dagger\otimes I_\B)\rho_{\R\B}^{\alpha}
                    (UV_jW_z^\dagger\otimes I_\B)}
 {dI_\M\otimes
  \Tr_\M[(V_j^\dagger U^\dagger\otimes I_\B)\rho_{\R\B}^{\alpha}(UV_j\otimes I_\B)]}
 \right\}_{z=1}^{d^2}.
 \label{eq:code-povm}
\end{equation}
The Weyl identity \eqref{eq:twirl} makes the numerators sum to the denominator. Linearity of either quotient therefore proves normalization on its support. On the kernel of the $\B$-marginal in the denominator, complete each test by $I_\M/d^2$. Every hypothesis has zero weight there because $\rho_{\R\B}$ and $\rho_{\R\B}^{\alpha}$ have the same support.

Let $T_{\M\B}^{U,j}$ be the completed test for $W_1=I_\M$. Its Weyl conjugates are the remaining tests, and Theorem~\ref{thm:duality} supplies the corresponding decoder. Write its entanglement-generating error as
\begin{equation}
 \eps^{U,j}
 =1-\Tr[\theta_{\M\B}^{U,j}T_{\M\B}^{U,j}].
 \label{eq:code-error}
\end{equation}
Thus $(U,j)$ specifies the encoding and decoder, while $z$ labels the discrimination outcomes. All tests are constructed from the source power $\rho_{\R\B}^{\alpha}$ before compression.

For the auxiliary Weyl ensemble defined in Appendix~\ref{app:petz},
the symmetric and logarithmic-derivative measurements are the
ordinary and averaged rotated Petz recovery measurements,
respectively \cite{BDL,SBT,JRSWW}.
Theorem~\ref{thm:duality} converts these measurements into decoding
channels acting on Bob's system $\B$.
Appendix~\ref{app:petz} gives their explicit formulas.
The following proof applies to either choice.

\subsection{Error analysis} \label{sec:proof}

\begin{proof}[Proof of Theorem~\ref{thm:coding}]

Fix $U$. Let $\mathsf J$ be a finite register with one basis vector for each encoding index, and define the isometry
\begin{equation}
 L_U:\R\longrightarrow\mathsf J\M,
 \qquad
 L_U=\sum_j\sqrt{\frac{|\R|w_j}{d}}\,
       \ket j_{\mathsf J}\otimes V_j^\dagger U^\dagger,
 \qquad L_U^\dagger L_U=I_\R.
 \label{eq:encoding-isometry}
\end{equation}
The last identity is \eqref{eq:resolution}. Applied to the source purification, measurement of $\mathsf J$ produces exactly \eqref{eq:encoding-probability}--\eqref{eq:encoding}. In the proof we retain this register and use the auxiliary state
\begin{equation}
 \rho_{\mathsf J\M\B}^{U}=L_U\rho_{\R\B}L_U^\dagger,
 \qquad
 (\rho_{\mathsf J\M\B}^{U})^\alpha
 =L_U\rho_{\R\B}^{\alpha}L_U^\dagger.
 \label{eq:coherent-source}
\end{equation}
Identity operators on $\B$ are implicit. The power identity follows immediately by applying $L_U$ to an orthonormal eigenbasis of $\rho_{\R\B}$.

Cauchy--Schwarz inequality and the dimension bound $X_{\M\B}\leq dI_\M\otimes\Tr_\M X_{\M\B}$ give
\begin{equation}
 (\rho_{\mathsf J\M\B}^{U})^\alpha
 \leq\frac{|\R|}{d}\bigoplus_j
     V_j^\dagger U^\dagger\rho_{\R\B}^{\alpha}UV_j
 \leq |\R|\bigoplus_j
 \left(I_\M\otimes\Tr_\M[V_j^\dagger U^\dagger\rho_{\R\B}^{\alpha}UV_j]\right).
 \label{eq:block-comparison}
\end{equation}
Indeed, the first inequality is the quadratic-form estimate
\[
 \left\|\sum_j\sqrt{w_j}\,\rho_{\R\B}^{\alpha/2}UV_jx_j\right\|^2
 \leq\sum_j\left\|\rho_{\R\B}^{\alpha/2}UV_jx_j\right\|^2.
\]
Apply the positive linear quotient with the last operator in \eqref{eq:block-comparison} as denominator. Its quotient with the middle operator is the direct sum of the uncompleted identity tests, and kernel completion adds positive terms. Consequently,
\begin{equation}
 \frac{(\rho_{\mathsf J\M\B}^{U})^\alpha}
 {|\R|\displaystyle\bigoplus_j
  \left(I_\M\otimes\Tr_\M[V_j^\dagger U^\dagger\rho_{\R\B}^{\alpha}UV_j]\right)}
 \leq\bigoplus_jT_{\M\B}^{U,j}.
 \label{eq:source-comparison}
\end{equation}

The diagonal blocks of $\rho_{\mathsf J\M\B}^{U}$ are $p_j(U)\theta_{\M\B}^{U,j}$. Thus
\begin{align}
 \sum_jp_j(U)\eps^{U,j}
 &=1-\Tr\!\left[\rho_{\mathsf J\M\B}^{U}\bigoplus_jT_{\M\B}^{U,j}\right]\notag\\
 &\overset{\textnormal{(a)}}{\leq}\Tr\!\left[\rho_{\mathsf J\M\B}^{U}
 \frac{|\R|\displaystyle\bigoplus_j
  \left(I_\M\otimes\Tr_\M[V_j^\dagger U^\dagger\rho_{\R\B}^{\alpha}UV_j]\right)
  -(\rho_{\mathsf J\M\B}^{U})^\alpha}
 {|\R|\displaystyle\bigoplus_j
  \left(I_\M\otimes\Tr_\M[V_j^\dagger U^\dagger\rho_{\R\B}^{\alpha}UV_j]\right)}\right]\notag\\
 &\overset{\textnormal{(b)}}{\leq} c_\alpha\Tr\!\left[(\rho_{\mathsf J\M\B}^{U})^\alpha
 \left(|\R|\bigoplus_j
  \left(I_\M\otimes\Tr_\M[V_j^\dagger U^\dagger\rho_{\R\B}^{\alpha}UV_j]\right)
  -(\rho_{\mathsf J\M\B}^{U})^\alpha\right)^{\frac{1-\alpha}{\alpha}}\right]\notag\\
 &\overset{\textnormal{(c)}}{\leq} c_\alpha\Tr\!\left[\rho_{\R\B}^{\alpha}
 \left(\frac{|\R|^2}{d}\sum_jw_j
  (UV_jV_j^\dagger U^\dagger)\otimes
  \Tr_\M[V_j^\dagger U^\dagger\rho_{\R\B}^{\alpha}UV_j]
  -\rho_{\R\B}^{\alpha}\right)^{\frac{1-\alpha}{\alpha}}\right].
 \label{eq:tilt-before-return}
\end{align}
The first inequality (a) uses \eqref{eq:source-comparison}; 
(b) follows from the tilting inequality, Lemma~\ref{lemm:tilting} of Appendix~\ref{app:tilting}, and \eqref{eq:coherent-source}--\eqref{eq:block-comparison}. For the last inequality (c), apply operator Jensen inequality \cite{HP} in the form
\[
 L_U^\dagger X^{\frac{1-\alpha}{\alpha}}L_U
 \leq(L_U^\dagger XL_U)^{\frac{1-\alpha}{\alpha}},\qquad X\geq0.
\]
This applies because $L_U$ is an isometry and $x^{(1-\alpha)/\alpha}$ is operator concave for $\sfrac12\leq\alpha<1$.

Now average over $U$. A second application of operator concavity moves the Haar expectation inside the positive power in the last line of \eqref{eq:tilt-before-return}. For each fixed $j$, the projection $UV_jV_j^\dagger U^\dagger$ has the same Haar second moment. Appendix~\ref{app:moment}, together with $\sum_jw_j=1$, therefore gives
\begin{align}
 &\frac{|\R|^2}{d}\Exp_U\!\left[
 \sum_jw_j(UV_jV_j^\dagger U^\dagger)\otimes
 \Tr_\M[V_j^\dagger U^\dagger\rho_{\R\B}^{\alpha}UV_j]
 \right]-\rho_{\R\B}^{\alpha}\notag\\
 &\qquad=
 \frac{d|\R|-1}{|\R|^2-1}
 \left(|\R|I_\R\otimes\Tr_\R\rho_{\R\B}^{\alpha}
       -\rho_{\R\B}^{\alpha}\right).
 \label{eq:mean}
\end{align}
By substitution, we have proved \eqref{eq:operator-bound} for $\avgeps(d)=\Exp_U\sum_jp_j(U)\eps^{U,j}$.

To obtain the scalar bound, concavity of the positive power implies superadditivity of the homogeneous function
\[
 X\mapsto
 \left(\Tr[\rho_{\R\B}^{\alpha}X^{(1-\alpha)/\alpha}]\right)^{\alpha/(1-\alpha)},
 \qquad X\geq0.
\]
To see this, normalize two operators to have value one and apply concavity with weights proportional to their original values; zero values follow by continuity. Apply superadditivity to $\rho_{\R\B}^{\alpha}$ and
$|\R|I_\R\otimes\Tr_\R\rho_{\R\B}^{\alpha}-\rho_{\R\B}^{\alpha}$.
The first has value $1$, while their sum has value $|\R|\,2^{-I_\alpha(\R\rangle\B)_\rho}$. Therefore
\begin{equation}
 \Tr\!\left[\rho_{\R\B}^{\alpha}
 \left(|\R|I_\R\otimes\Tr_\R\rho_{\R\B}^{\alpha}
                   -\rho_{\R\B}^{\alpha}\right)^{\frac{1-\alpha}{\alpha}}\right]
 \leq\left(|\R|\,2^{-I_\alpha(\R\rangle\B)_\rho}-1\right)^{\frac{1-\alpha}{\alpha}}.
 \label{eq:trace-correction}
\end{equation}
This proves \eqref{eq:scalar-bound}. Dropping the subtractive term and using
$|\R|(d|\R|-1)/(|\R|^2-1)\leq d$ gives \eqref{eq:exponential-bound}.

For $d\geq|\R|$, choose an isometry $E:\R\to\M$ and prepare $(E\otimes I_\A)\ket\phi_{\R\A}$. The channel input is exactly $\phi_\A$. Denote the embedded output by $\rho_{\M\B}$. Its $d^2$ Weyl hypotheses have the POVM
\begin{equation}
 \left\{
 \frac{(W_z\otimes I_\B)\rho_{\M\B}^{\alpha}(W_z^\dagger\otimes I_\B)}
 {dI_\M\otimes\Tr_\R\rho_{\R\B}^{\alpha}}
 \right\}_{z=1}^{d^2},
 \label{eq:large-target-povm}
\end{equation}
with uniform completion on the marginal kernel. Normalization follows from the Weyl sum. Theorem~\ref{thm:duality} gives its decoder. Since
$dI_\R\otimes\Tr_\R\rho_{\R\B}^\alpha\geq\rho_{\R\B}^\alpha$,
tilting applies directly. Restricting the resulting trace to the embedded reference proves \eqref{eq:large-target-operator}. The same scalar argument, with $d$ in place of $|\R|$, gives \eqref{eq:large-target-scalar} and \eqref{eq:exponential-bound}.

Finally, some preparation and its decoder have error no greater than the ensemble average. Fix them before communication. The prescribed input $\phi_\A$ is the ensemble average and need not be the marginal of the selected preparation.
\end{proof}

\begin{remark}
For either quotient, the identity effect in \eqref{eq:large-target-povm} is $|\R|/d$ times the embedded identity effect at target $|\R|$, on the source support. Equivalently, one may decode to dimension $|\R|$ and then embed the output using the complex conjugate of $E$. Their errors obey
\begin{equation}
 \eps(d)=1-\frac{|\R|}{d}+\frac{|\R|}{d}\eps(|\R|).
 \label{eq:embedding-identity}
\end{equation}
Substituting the bound at target $|\R|$ gives an optional complementary estimate, which may improve the direct bound at a fixed $\alpha$. Kernel completions do not affect this identity, and $|\R|=1$ is included since its target error is zero.
\end{remark}

\subsection{Exact noiseless error and the sharp coefficient}
\label{sec:tightness}

Consider the identity channel on $\A$ with uniform input and minimal reference $|\R|=|\A|$. Its output is maximally entangled, so $\rho_{\R\B}^\alpha=\rho_{\R\B}$ and $I_\alpha(\R\rangle\B)_\rho=\log_2|\A|$. For $2\leq d\leq|\A|$, the operator
\[
 |\R|I_\R\otimes\Tr_\R\rho_{\R\B}-\rho_{\R\B}
 =I_{\R\B}-\rho_{\R\B}
\]
annihilates the source. Both \eqref{eq:operator-bound} and \eqref{eq:scalar-bound} give zero error, including when $d$ does not divide $|\A|$. Each selected preparation is maximally entangled of dimension $d$.

For $d>|\A|$, the exact optimal error is
\begin{equation}
 \min_{\psi,\D}\eps(\psi,\D;\id_\A)=1-\frac{|\A|}{d}.
 \label{eq:identity-error}
\end{equation}
The discrimination perspective gives a short converse. For any preparation, its $d^2$ Weyl hypotheses are density operators on a space of dimension $d|\A|$, and each is bounded by the identity. Every POVM therefore has average success at most
\[
 \frac1{d^2}\sum_{z=1}^{d^2}\Tr[\Pi_{\M\B}^z]
 =\frac{|\A|}{d}.
\]
Theorem~\ref{thm:duality} gives the same upper bound on squared recovery fidelity. For either quotient, our POVM \eqref{eq:large-target-povm} attains it: its identity effect on the source is $(|\A|/d)\rho_{\M\B}$. Equivalently, \eqref{eq:embedding-identity} embeds a decoder with zero error. 

The direct bounds \eqref{eq:large-target-operator}--\eqref{eq:large-target-scalar} separately test the sharp prefactor. They reduce to
\begin{equation}
 \eps(d)\leq
 c_\alpha\left(\frac d{|\A|}-1\right)^{\frac{1-\alpha}{\alpha}}.
 \label{eq:identity-bound}
\end{equation}
The choice
\begin{equation}
 \alpha=\frac d{d+|\A|}
 \quad\text{gives}\quad
 c_\alpha\left(\frac d{|\A|}-1\right)^{\frac{1-\alpha}{\alpha}}
 =1-\frac{|\A|}{d}.
 \label{eq:identity-equality}
\end{equation}
Thus this bound too meets the converse exactly.

Moreover, its universal coefficient cannot be reduced at any fixed order. For $\alpha\in(\sfrac12,1)$, choose integer dimensions with $|\A|/d$ approaching $(1-\alpha)/\alpha$. Then
\[
 \frac{1-|\A|/d}{(d/|\A|-1)^{(1-\alpha)/\alpha}}
 \longrightarrow c_\alpha.
\]
At $\alpha=\sfrac12$, let $|\A|/d$ tend to one. Necessity of this coefficient in \eqref{eq:large-target-scalar} is unconditional; its universal sufficiency is Lemma~\ref{lemm:tilting}.

\subsection{Entanglement transmission and strong subspace transmission}
\label{sec:transmission}

Let $\mathsf M'$ be a copy of $\M$. An entanglement-transmission code consists of an encoding channel $\mathscr E_{\mathsf M'\to\A}$ and a decoder $\D_{\B\to\Mh}$, with error
\begin{equation}
 \varepsilon_{\rm ET}(\mathscr E,\D;\N)
 \coloneqq P^2\!\left((\id_\M\otimes(\D\circ\N\circ\mathscr E))(\Phi_{\M\mathsf M'}),\Phi_{\M\Mh}\right).
 \label{eq:et-error}
\end{equation}
Every entanglement-generation code of dimension $d$ and error $\eps(d)$ yields an entanglement-transmission code of the same dimension satisfying \cite[Lemma  14.7]{Wilde_book}
\begin{equation}
 \varepsilon_{\rm ET}(d)\leq 2\eps(d)-\eps(d)^2\leq 2\eps(d).
 \label{eq:eg-to-et}
\end{equation}
We give a short proof here.
Write the encoding as
$\ket\psi_{\M\A}=\sum_{m=0}^{d-1}\ket m_\M\otimes X\ket m_{\mathsf M'}$, where $\Tr [X^\dagger X]=1$.
Rotate a Kraus representation $\{K_i\}_i$ of $\D\circ\N$ so that
\[
 |\Tr K_0X|^2=d(1-\eps(d)).
\]
Since $K_0$ is a contraction, Cauchy--Schwarz inequality gives
$\Tr|K_0|\geq\Tr K_0^\dagger K_0\geq d(1-\eps(d))$.
Choose the adjoint of the polar partial isometry of $K_0$ as one Kraus operator of $\mathscr E$, and complete it to a channel. The corresponding contribution to its entanglement fidelity is $(\Tr|K_0|)^2/d^2\geq(1-\eps(d))^2$, proving \eqref{eq:eg-to-et}. If $d\leq|\A|$, extend that partial isometry to an isometry instead.

Strong subspace transmission requires the same code to preserve every pure input, including its correlations with a reference: its squared purified-distance error is
\begin{equation}
 \varepsilon_{\rm ST}(\mathscr E,\D;\N)
 \coloneqq\max_{\varphi_{\M\mathsf M'}}
 P^2\!\left((\id_\M\otimes(\D\circ\N\circ\mathscr E))(\varphi_{\M\mathsf M'}),\varphi_{\M\Mh}\right),
 \label{eq:st-error}
\end{equation}
where the maximum is over pure states and the target identifies $\mathsf M'$ with $\Mh$.
By \cite[Lemma~24]{BCY}, for every $\eta\in(0,1)$ an entanglement-transmission code of dimension $d$ can be restricted to a subspace of dimension $d_\eta\geq\lceil(1-\eta)d\rceil$ such that
\begin{equation}
 \varepsilon_{\rm ST}(d_\eta)
 \leq\frac{\varepsilon_{\rm ET}(d)}{\eta}
 \leq\frac{2\eps(d)-\eps(d)^2}{\eta}.
 \label{eq:eg-to-st}
\end{equation}
In particular, retaining at least half the dimension costs at most a factor $2$ in purified distance relative to the generation code. These constant losses preserve the achievable error exponent. For any ensemble bound above, first select a code with $\eps(d)\leq\avgeps(d)$ and then apply these conversions.

\section{Hypothesis-testing bounds}
\label{sec:small-deviations}

%At order $\alpha=1$, a single fixed isometry and Haar rotation suffice. Both quotient conventions in \eqref{eq:quotient} remain available. 

Apart from the R\'enyi-type bounds established in Section~\ref{sec:coding}, we show that the hypothesis-testing bound developed by the author for classical communication over quantum channels \cite{ChengSimple} extends to entanglement generation as well.
We note that a concurrent and independent work by Can and Wilde \cite{CW} also obtained the same bound \eqref{eq:haar-wedge} below.
For completeness, we present an alternative proof here.

Write 
\begin{align}
A\wedge B
&\coloneqq \argmax_{H=H^\dagger} \left\{ \Tr[H] : H\leq A, H\leq B \right\}
\label{eq:def:minimum}
\\
&= \frac{A+B-|A-B|}{2}
\end{align}
for the noncommutative minimum.
Its trace gives an operational meaning to minimum error of binary hypothesis testing:
\begin{align} \label{eq:eq:minimum_test}
	\Tr[A\wedge B]
	= \inf_{0\leq T \leq \I} \Tr[A(I-T)] + \Tr[BT].
\end{align}

\begin{theorem}[Hypothesis-testing bounds]
	\label{thm:binary}
	For either quotient in \eqref{eq:quotient}, and $2\leq d\leq|\R|$, the Haar ensemble
	defined below has average channel input $\phi_\A$.
	Its average entanglement-generating error satisfies
	\begin{align}
		\avgeps(d)
		&\leq
		\Tr\!\left[\rho_{\R\B}
		\frac{\frac{d|\R|-1}{|\R|^2-1}
			(|\R|I_\R\otimes\rho_\B-\rho_{\R\B})}
		{\rho_{\R\B}+\frac{d|\R|-1}{|\R|^2-1}
			(|\R|I_\R\otimes\rho_\B-\rho_{\R\B})}\right]
		\label{eq:haar-quotient}\\
		&\leq
		\Tr\!\left[
		\rho_{\R\B}\wedge
		\frac{d|\R|-1}{|\R|^2-1}
		(|\R|I_\R\otimes\rho_\B-\rho_{\R\B})
		\right].
		\label{eq:haar-wedge}
	\end{align}
	For $d\geq|\R|$, the direct preparation and decoder in
	\eqref{eq:large-target-povm}, at $\alpha=1$, satisfy
	\begin{align}
		\eps(d)
		&\leq
		\Tr\!\left[\rho_{\R\B}
		\frac{dI_\R\otimes\rho_\B-\rho_{\R\B}}
		{dI_\R\otimes\rho_\B}\right]
		\label{eq:large-quotient}\\
		&\leq
		\Tr\!\left[
		\rho_{\R\B}\wedge
		(dI_\R\otimes\rho_\B-\rho_{\R\B})
		\right].
		\label{eq:large-wedge}
	\end{align}
	Consequently, every integer $d\geq1$ admits a code with
	\begin{equation}
		\eps(d)\leq
		\Tr[\rho_{\R\B}\wedge dI_\R\otimes\rho_\B].
		\label{eq:universal-wedge}
	\end{equation}
\end{theorem}
\begin{proof}
	For $2\leq d\leq|\R|$, fix an isometry $V:\M\to\R$ and let
	$\mathrm dU$ be normalized Haar measure on the unitaries of $\R$.
	Choose the encoding
	\begin{equation}
		p(U)=\frac{|\R|}{d}
		\Tr[V^\dagger U^\dagger\phi_\R UV],
		\qquad
		\ket{\psi^U}_{\M\A}
		=
		\frac{(V^\dagger U^\dagger\otimes I_\A)\ket\phi_{\R\A}}
		{\sqrt{\Tr[V^\dagger U^\dagger\phi_\R UV]}}.
		\label{eq:haar-encoding}
	\end{equation}
%	Zero-density preparations are omitted.
	Here, $p(U)$ is a probability density relative to $\mathrm dU$.
	The first Haar moment
	\[
	\Exp_U[UVV^\dagger U^\dagger]
	=\frac d{|\R|}I_\R
	\]
	gives $\Exp_U[p(U)]=1$ and
	$\Exp_U[p(U)\psi_\A^U]=\phi_\A$.
	
	For each encoding, use the Weyl discrimination POVM whose equivalent test is
	\begin{equation}
		T_{\M\B}^U
		=
		\frac{V^\dagger U^\dagger\rho_{\R\B}UV}
		{dI_\M\otimes
			\Tr_\M[V^\dagger U^\dagger\rho_{\R\B}UV]}.
		\label{eq:haar-petz-effect}
	\end{equation}
	Isometries and unitaries act trivially on $\B$.
	Complete every Weyl effect uniformly on the unused marginal kernel,
	which has zero overlap with the output.
	Theorem~\ref{thm:duality} supplies the corresponding decoder.
	
	We first average the success probability.
	For the symmetric quotient, squared-norm Jensen and Jensen's
	operator inequality for $u\mapsto u^{-1/2}$ \cite{HP} give
	\begin{align}
		1-\avgeps(d)
		&=\frac{|\R|}{d^2}\Exp_U
		\left\|\rho_{\R\B}^{1/2}
		\left[
		UVV^\dagger U^\dagger\otimes
		\left(\Tr_\M[V^\dagger U^\dagger\rho_{\R\B}UV]\right)^{-1/2}
		\right]\rho_{\R\B}^{1/2}\right\|_2^2
		\notag\\
		&\geq
		\Tr\!\left[\rho_{\R\B}
		\frac{\rho_{\R\B}}
		{\displaystyle
			\frac{|\R|^2}{d}\Exp_U\!\left[
			UVV^\dagger U^\dagger\otimes
			\Tr_\M[V^\dagger U^\dagger\rho_{\R\B}UV]
			\right]}
		\right].
		\label{eq:binary-averaged-success}
	\end{align}
	Indeed, the first Haar moment makes the normalized averaging map
	\[
	(X_U)_U\longmapsto
	\frac{|\R|}{d}\Exp_U[UVX_UV^\dagger U^\dagger]
	\]
	unital, so operator Jensen applies.
	After sandwiching by $\rho_{\R\B}^{1/2}$, we also use that
	$\Tr [H^2]$ is increasing in the operator order for $H\geq0$.
	
	For the logarithmic-derivative quotient, its integral representation
	gives
	\[
	1-\avgeps(d)
	=
	\frac{|\R|}{d^2}\int_0^\infty\Exp_U
	\left\|\rho_{\R\B}^{1/2}
	\left[
	UVV^\dagger U^\dagger\otimes
	\left(\Tr_\M[V^\dagger U^\dagger\rho_{\R\B}UV]
	+tI_\B\right)^{-1}
	\right]\rho_{\R\B}^{1/2}\right\|_2^2\,\mathrm dt.
	\]
	Applying the same two Jensen inequalities, now for $u\mapsto u^{-1}$,
	and rescaling the integration variable by $|\R|$ proves
	\eqref{eq:binary-averaged-success} for this quotient as well.
	For singular marginals, first add $\delta I_\B$ before applying
	Jensen and then let $\delta\downarrow0$ after sandwiching by the
	source. The resulting limits follow by monotone convergence.
	
	The twirling identity \eqref{eq:projector-moment} in
	Appendix~\ref{app:moment} evaluates the averaged denominator:
	\begin{align}
		\frac{|\R|^2}{d}\Exp_U\!\left[
		UVV^\dagger U^\dagger\otimes
		\Tr_\M[V^\dagger U^\dagger\rho_{\R\B}UV]\right]
		=
		\rho_{\R\B}
		+\frac{d|\R|-1}{|\R|^2-1}
		\left(|\R|I_\R\otimes\rho_\B-\rho_{\R\B}\right).
		\label{eq:binary-averaged-denominator}
	\end{align}
	This operator dominates $\rho_{\R\B}$ by the dimension bound.
	Since either quotient is linear in its numerator and the quotient
	of an operator with itself is its support projection, subtracting
	\eqref{eq:binary-averaged-success} from one proves
	\eqref{eq:haar-quotient}.
	
	The binary-testing inequality
	\cite[Eq.~(A1)]{ChengSimple}, \cite[Eqs.~(68)--(69)]{CL},
	\begin{equation}
		\Tr\!\left[A\frac{B}{A+B}\right]
		\leq\Tr[A\wedge B],
		\qquad A,B\geq0,
		\label{eq:binary-quotient}
	\end{equation}
	then gives \eqref{eq:haar-wedge}.
	
	For $d\geq|\R|$, substituting $\alpha=1$ into the direct decoder
	\eqref{eq:large-target-povm} gives equality in
	\eqref{eq:large-quotient}. Its numerator is positive by the
	dimension bound, so \eqref{eq:binary-quotient} also proves
	\eqref{eq:large-wedge}.
	
	Finally, the trace minimum is increasing
	in each argument by \eqref{eq:eq:minimum_test}. Dropping the subtractive term and using
	$|\R|(d|\R|-1)/(|\R|^2-1)\leq d$ for $d\leq|\R|$
	gives \eqref{eq:universal-wedge}.
	At least one preparation and its decoder attain the ensemble
	bound; fix this realization before communication.
	The case $d=1$ has zero error.
\end{proof}

\begin{remark}[Entanglement transmission and second-order achievability]
	Beigi, Datta, and Leditzky \cite{BDL} established one-shot and
	second-order achievability bounds for entanglement transmission
	using random code subspaces and Petz recovery.
	Can and Wilde \cite[Lemma~4.4]{CW} showed that an
	entanglement-generation code can be converted into an
	entanglement-transmission code at the same infidelity, with only
	a constant loss in the logarithm of the target dimension at
	fixed error. 
	Applying this conversion to the codes obtained from
	\eqref{eq:universal-wedge} and the standard second-order expansion \cite{LiSecondOrder, TomamichelHayashi2013}, yields the same
	achievable second-order rate, for product
	sources over a memoryless channel
	\cite[Theorem~6.3]{CW}.
\end{remark}

\begin{remark}[A strengthened hypothesis-testing bound]
For $A,B\geq0$ with $\Tr[A+B]>0$, the stronger trace inequality is
\begin{equation}
 \Tr\!\left[A\frac{B}{A+B}\right]
 \leq\Tr[A\wedge B]
 -\frac{\bigl(\Tr[A\wedge B]\bigr)^2}{\Tr[A+B]}.
 \label{eq:stronger-binary-quotient}
\end{equation}
The symmetric-quotient form is \cite[(A2)]{ChengSimple}.
For the logarithmic-derivative quotient, the same strengthened bound follows from the proof of \cite[Lemma~4, Appendix~B]{BeigiTomamichel24}.
For the two operators used above, $\Tr[A+B]=d|\R|$. Thus \eqref{eq:haar-wedge} strengthens to
\begin{align}
 \avgeps(d)
 &\leq\Tr\!\left[\rho_{\R\B}\wedge
 \frac{d|\R|-1}{|\R|^2-1}
 (|\R|I_\R\otimes\rho_\B-\rho_{\R\B})\right]\notag\\
 &\quad-\frac1{d|\R|}
 \left(\Tr\!\left[\rho_{\R\B}\wedge
 \frac{d|\R|-1}{|\R|^2-1}
 (|\R|I_\R\otimes\rho_\B-\rho_{\R\B})\right]\right)^2.
 \label{eq:strengthened-haar-wedge}
\end{align}
For $d\geq|\R|$, the same improvement holds with
$dI_\R\otimes\rho_\B-\rho_{\R\B}$ as the second argument of each trace minimum; its squared term has the same denominator $d|\R|$.
\end{remark}

\section{Conclusions}
\label{sec:consequences}

In this paper, we established one-shot coding bounds for entanglement generation through state discrimination. The approach provides explicit decoders built from Petz recovery maps \cite{BDL,CW}, whereas the decoupling construction of \cite{HHWY} establishes the existence of a decoder through Uhlmann's theorem. It extends the decoding error analysis developed for classical communication over quantum channels \cite{ChengSimple,CL} to quantum communication. Both constructions use random coding to establish the existence of a suitable fixed code. Although our analysis uses Haar-random unitaries, the same bounds hold for any exact unitary $2$-design, since only the first two moments enter the proof \cite{HHWY,DCEL}.

When both errors are expressed as squared purified distance, the resulting achievable random-coding exponent agrees with that previously obtained through decoupling in \cite{BCY}. Our approach thereby provides an alternative derivation of this exponent through state discrimination and an explicit analysis of the decoding error.

\appendix
\section{Auxiliary lemmas}
\label{app:auxiliary}

\subsection{Tilting inequality}
\label{app:tilting}

We employ the following sharp tilting inequality for either quotient.

\begin{samepage}
\begin{lemm}[Sharp tilting inequality \cite{CL_sharptilting}]
\label{lemm:tilting}
For either convention in \eqref{eq:quotient}, $\alpha\in[\sfrac12,1)$, and positive semidefinite operators $A$ and $B$,
\begin{equation}
 \Tr\!\left[A\frac{B^\alpha}{A^\alpha + B^\alpha}\right]
 \leq c_\alpha\Tr\!\left[A^\alpha B^{1-\alpha}\right],
 \label{eq:tilting}
\end{equation}
where
\begin{equation}
 c_\alpha=
 \left(\frac{1-\alpha}{\alpha}\right)^{\frac{1-\alpha}{\alpha}}
 \left(\frac{2\alpha-1}{\alpha}\right)^{\frac{2\alpha-1}{\alpha}} \in \left[\sfrac12,1\right].
 \label{eq:constant}
\end{equation}
The endpoint is defined by continuity, i.e.~$c_{\sfrac12}=c_1 = 1$.
\end{lemm}
\end{samepage}

For the logarithmic-derivative quotient, the inequality proved in \cite[Proposition~2.3]{CL} via the operator layer cake theorem has the same form with a coefficient less than $1.102$.

\subsection{The Haar twirling identity}
\label{app:moment}

Let $P_\R$ be a Haar-random rank-$d$ projection on $\R$, with $|\R|\geq2$. For every operator $X_{\R\B}$,
\begin{align}
\Exp\!\left[P_\R\otimes\Tr_\R[(P_\R\otimes I_\B)X_{\R\B}]\right]
 =
 \frac{d(d|\R|-1)}{|\R|(|\R|^2-1)}
 I_\R\otimes\Tr_\R X_{\R\B}
 +\frac{d(|\R|-d)}{|\R|(|\R|^2-1)}X_{\R\B}.
 \label{eq:projector-moment}
\end{align}
This is the standard random-subspace twirling identity; see \cite[Lemma 2.2]{HHWY}, after multiplying the twirled restricted swap by the swap and contracting with $X_{\R\B}$. Applying it to each $UV_jV_j^\dagger U^\dagger$, then summing the weights, proves \eqref{eq:mean}. Applying it to the single projection $UVV^\dagger U^\dagger$ gives the denominator in Section~\ref{sec:small-deviations}.

\section{Petz recovery and the decoding measurement}
\label{app:petz}

The two quotients in \eqref{eq:quotient} give the ordinary and averaged rotated Petz recovery measurements for the same auxiliary Weyl ensemble. Theorem~\ref{thm:duality} converts these measurements into the decoding channels used in the paper.

Fix an encoding $(U,j)$ of positive probability and set
\begin{equation}
 \omega_{\M\B}
 =\frac{V_j^\dagger U^\dagger\rho_{\R\B}^{\alpha}UV_j}
 {\Tr[V_j^\dagger U^\dagger\rho_{\R\B}^{\alpha}UV_j]},
 \qquad \omega_\B=\Tr_\M\omega_{\M\B}.
 \label{eq:petz-auxiliary}
\end{equation}
For the single-isometry construction omit $j$, and for the direct construction use the normalized $\rho_{\M\B}^{\alpha}$. Scalar normalization cancels from both quotients. Let $\mathsf X$ record the $d^2$ Weyl labels. The classical--quantum preparation channel with uniform reference state is
\begin{align}
 \mathscr E(\xi_{\mathsf X})
 &=\sum_{z=1}^{d^2}\bra z\xi_{\mathsf X}\ket z\,
 (W_z\otimes I_\B)\omega_{\M\B}(W_z^\dagger\otimes I_\B),
 \qquad \sigma_{\mathsf X}=\frac{I_{\mathsf X}}{d^2},
 \label{eq:weyl-preparation}\\
 \mathscr E(\sigma_{\mathsf X})
 &=\frac{I_\M}{d}\otimes\omega_\B.
 \label{eq:weyl-average}
\end{align}
The last equality is the Weyl twirl.

For a channel $\mathscr E$ and reference state $\sigma$, the rotated Petz map is
\begin{equation}
 \mathcal R_{\sigma,\mathscr E}^{[t]}(\xi)
 =\sigma^{\frac{1+it}{2}}\mathscr E^*\!\left(
   \mathscr E(\sigma)^{-\frac{1+it}{2}}\xi
   \mathscr E(\sigma)^{-\frac{1-it}{2}}
   \right)\sigma^{\frac{1-it}{2}},
 \qquad t\in\mathbb R.
 \label{eq:rotated-petz}
\end{equation}
The ordinary Petz map is $\mathcal R^{[0]}_{\sigma,\mathscr E}$. Its averaged rotated version is $\int_{\mathbb R}\mathcal R^{[t]}_{\sigma,\mathscr E}\,\mathrm d\mu(t)$, where the averaging measure is
$\mathrm d\mu(t)=\pi\,\mathrm dt/[2(\cosh(\pi t)+1)]$ \cite[Corollary 4.2]{SBT}; see also \cite{JRSWW}.
For \eqref{eq:weyl-preparation}, the uniform prior makes the outer rotations cancel. Substitution into \eqref{eq:rotated-petz} yields the measurement channel
\begin{align}
 \mathcal R_{\sigma_{\mathsf X},\mathscr E}^{[t]}(\xi_{\M\B})
 &=\sum_{z=1}^{d^2}\Tr[\xi_{\M\B}\Pi_z^{[t]}]\ket z\bra z,
 \notag\\
 \Pi_z^{[t]}
 &=\frac1d
 (W_z\otimes\omega_\B^{-\frac{1-it}{2}})\omega_{\M\B}
 (W_z^\dagger\otimes\omega_\B^{-\frac{1+it}{2}}).
 \label{eq:rotated-effects}
\end{align}
At $t=0$, these are exactly the symmetric-quotient effects. For the average, the identity \cite[Lemma 3.4]{SBT}
\begin{equation}
 \mathrm D\ln[X](Y)
 =\int_{\mathbb R}
 X^{-\frac{1-\mathrm{i}t}{2}}Y X^{-\frac{1+\mathrm{i}t}{2}}\,\mathrm d\mu(t)
 \label{eq:log-rotation-identity}
\end{equation}
gives
\begin{equation}
 \int_{\mathbb R}\Pi_z^{[t]}\,\mathrm d\mu(t)
 =\mathrm D\ln[{dI_\M\otimes\omega_\B}]\!\left(
 (W_z\otimes I_\B)\omega_{\M\B}(W_z^\dagger\otimes I_\B)\right),
 \label{eq:averaged-effects}
\end{equation}
which is exactly the logarithmic-derivative measurement.

Finally, \eqref{eq:decoder-action} gives the corresponding
decoding channels explicitly:
\begin{align}
 \D_0(\xi_\B)
 &=\left(\Tr_\B\!\left[
 \omega_{\M\B}
 (I_\M\otimes\omega_\B^{-\sfrac12}\xi_\B\omega_\B^{-\sfrac12})
 \right]\right)^{\mathsf T},
 \label{eq:petz-decoder-explicit}\\
 \D_{\rm int}(\xi_\B)
 &=\left(\Tr_\B\!\left[
 \omega_{\M\B}(I_\M\otimes\mathrm D\ln[\omega_\B](\xi_\B))
 \right]\right)^{\mathsf T}\notag\\
 &=\int_{\mathbb R}
 \D_0\!\left(\omega_\B^{-\mathrm{i}t/2}\xi_\B\omega_\B^{\mathrm{i}t/2}\right)
 \,\mathrm d\mu(t).
 \label{eq:integral-decoder-explicit}
\end{align}
Thus the integral decoder corresponds exactly to the averaged rotated Petz map for the preparation channel \eqref{eq:weyl-preparation}. All powers act on the relevant supports. The channel completions send the orthogonal input subspace to $I_{\Mh}/d$ and discard cross terms, as prescribed by the uniform POVM completion in the main text.

\urlstyle{same}

\end{document}